\documentclass[5p,times]{article}%%elsarticle}

\usepackage[english]{babel}
\usepackage[utf8]{inputenc}
\usepackage{amsmath}
\usepackage{amsfonts}
\usepackage{amsthm}
\usepackage{amssymb}
\usepackage{fancyhdr}
\usepackage{xifthen}
\usepackage{tabularx,lipsum,environ,amsmath,amssymb}
\usepackage[T1]{fontenc}
\usepackage{todonotes}
\usepackage{hyperref}
\usepackage{tikz}
\usepackage{aliascnt}
\usepackage{tikz-qtree,tikz-qtree-compat}
\usepackage{url}
\usepackage{graphicx}
\usepackage{svg}
\usepackage{xcolor}

\addto\extrasenglish{%
	\def\equationautorefname~#1\null{Equation~(#1)\null}%
}

\usetikzlibrary{shapes,arrows,matrix,positioning,calc,positioning,automata,shadows,fit}

\newcommand{\msize}{2em}
\tikzstyle{redblock} = [draw, fill=red, rectangle, minimum width=\msize, minimum height=\msize]
\tikzstyle{blueblock} = [draw, fill=blue, rectangle, minimum width=\msize, minimum height=\msize]
\tikzstyle{grayblock} = [draw, fill=gray, rectangle, minimum width=\msize, minimum height=\msize]
\tikzstyle{whiteblock} = [draw, fill=white, rectangle, minimum width=\msize, minimum height=\msize]

\makeatletter
\newcommand{\problemtitle}[1]{\gdef\@problemtitle{#1}}% Store problem title
\newcommand{\probleminput}[1]{\gdef\@probleminput{#1}}% Store problem input
\newcommand{\problemquestion}[1]{\gdef\@problemquestion{#1}}% Store problem question
\NewEnviron{problem}{
  \problemtitle{}\probleminput{}\problemquestion{}% Default input is empty
  \BODY% Parse input
  \par\addvspace{.5\baselineskip}
  \noindent
  \begin{tabularx}{\linewidth}{@{\hspace{\parindent}} l X c}
    \multicolumn{2}{@{\hspace{\parindent}}l}{\textsc{\@problemtitle}} \\% Title
    \textbf{Input:} & \@probleminput \\% Input
    \textbf{Question:} & \@problemquestion% Question
  \end{tabularx}
  \par\addvspace{.5\baselineskip}
}

\newtheorem{theorem}{Theorem}[section]

\newaliascnt{lemma}{theorem}
\newtheorem{lemma}[lemma]{Lemma}
\aliascntresetthe{lemma}

\newaliascnt{proposition}{theorem}
\newtheorem{proposition}[proposition]{Proposition}
\aliascntresetthe{proposition}

\newaliascnt{corollary}{theorem}

\aliascntresetthe{corollary}

\newaliascnt{definition}{theorem}
\newtheorem{definition}[definition]{Definition}
\aliascntresetthe{definition}

\newcommand{\NP}{$\mathcal{NP}$}

\newcommand{\mbpt}{\textsc{$\epsilon$-Matrix \allowbreak Bipartition} }
\newcommand{\gb}{\textsc{$\epsilon$-Graph \allowbreak Bisection} }
\newcommand{\geb}{\textsc{$\epsilon$-Graph \allowbreak Edge-\allowbreak Bisection} }

\newcommand{\sgeb}{\textsc{GEB}_{\epsilon} }
\newcommand{\sgb}{\textsc{GB}_{\epsilon} }

\makeatletter
\newsavebox{\measure@tikzpicture}
\NewEnviron{scaletikzpicturetowidth}[1]{%
  \def\tikz@width{#1}%
  \def\tikzscale{1}\begin{lrbox}{\measure@tikzpicture}%
  \BODY
  \end{lrbox}%
  \pgfmathparse{#1/\wd\measure@tikzpicture}%
  \edef\tikzscale{\pgfmathresult}%
  \BODY
}
\makeatother

\definecolor{mgreen}{rgb}{0,0.5,0}

\begin{document}

	\title{An improved exact algorithm and an NP-completeness proof for sparse matrix bipartitioning}

	\author{Timon E. Knigge \and Rob H. Bisseling}

	\maketitle

	\begin{abstract}
		We investigate sparse matrix bipartitioning -- a problem where
		we minimize the communication volume in parallel sparse
		matrix-vector multiplication. We prove, by reduction from
		graph bisection, that this problem is \NP-complete
		in the case where each side of the bipartitioning must
		contain a linear fraction of the nonzeros.

		We present an improved exact branch-and-bound algorithm
		which finds the minimum communication volume for a given
		matrix and
		maximum allowed imbalance. The algorithm is based on a
		maximum-flow bound and a packing bound, which extend
		previous matching and packing bounds.

		We implemented the algorithm in a new program called MP
		(Matrix Partitioner), which solved 839 matrices from the
		SuiteSparse collection to optimality, each within 24 hours
		of CPU-time. Furthermore, MP solved the difficult problem
		of the matrix \texttt{cage6} in about 3 days. The new
		program is on average more than ten times faster than the
		previous program MondriaanOpt.
		
		Benchmark results using the set of 839 optimally solved matrices
		show that combining the medium-grain/iterative refinement
		methods of the Mondriaan package with the hypergraph
		bipartitioner of the PaToH package produces sparse matrix
		bipartitionings on average within 10\% of the optimal
		solution.
	\end{abstract}
	%\begin{keyword}
	%	Bisection \sep Exact Algorithm \sep Maximum Flow \sep
	%	NP Complete \sep Partitioning \sep Sparse Matrix
	%\end{keyword}

	\section{Introduction}
\label{sec:intro}
Sparse matrix partitioning is important for the parallel solution of sparse linear systems by direct or iterative
methods. In iterative solvers, the basic kernel is the multiplication of a sparse matrix and a dense vector, the
SpMV operation. A good partitioning of the sparse matrix and the vector will balance the computation load in a
parallel SpMV by spreading the matrix nonzeros evenly over the parts assigned to the processors of the parallel
computer, while also leading to less communication of the vector components between the processors.

In the past decades, much effort has been spent on developing and improving heuristic partitioning methods. In
particular, hypergraph methods have been very successful because they model the communication volume (the total
number of data words sent) exactly. Two-dimensional (2D) partitioning methods are superior to 1D methods, since
they are more general and can split both the rows and columns of the matrix and hence in principle can provide
better solutions. Heuristic algorithms for hypergraph-based sparse matrix partitioning have been implemented in
the sequential software packages hMetis~\cite{karypis99b}, PaToH~\cite{catalyurek99},
Mondriaan~\cite{vastenhouw05}, KaHyPar~\cite{akhremtsev17}, and the parallel packages
Par$k$way~\cite{trifunovic08} and Zoltan~\cite{devine06}. The current state-of-the-art methods for 2D sparse
matrix partitioning are the fine-grain~\cite{catalyurek01} method and the medium-grain method~\cite{pelt14}. 

How good are the current methods and is it still worthwhile to improve them? To answer this question we need to
compare the quality of the outcome, i.e., the communication volume, to the optimal result. To enable such a
comparison, we need an exact algorithm that provides the minimum communication volume for a specfied maximum load
imbalance. The first exact algorithm for this problem (with two parts) was proposed by Pelt and
Bisseling~\cite{pelt15}, based on a branch-and-bound method.
This algorithm has been implemented in the program MondriaanOpt, included in the
Mondriaan package, version 4.2. As of today, 356 matrices from the SuiteSparse (formerly University of Florida)
sparse matrix collection~\cite{davis11} have been bipartitioned to optimality by
MondriaanOpt.\footnote{The solutions can be found at
\url{http://www.staff.science.uu.nl/~bisse101/Mondriaan/Opt/}.} Being able to increase the size of the
solution subset would be valuable for benchmarking heuristic partitioners, by providing more comparison data of
a more realistic size. Heuristic partitioners are aimed at large problems, although they will
encounter smaller problems after their inital splits.

Optimal partitionings are easiest to compute for splitting into two parts: the required computation time grows
quickly with a larger number of parts, as discussed by Pelt and Bisseling~\cite{pelt15}. Furthermore, heuristic
partitioners often are based on recursive bipartitioning, so that it is most important to gauge the quality of
the bipartitioner. (A notable exception is KaHyPar, which computes a direct $k$-way partitioning.) Therefore,
both the exact partitioner implemented in MondriaanOpt and the improved partioner  MP (for Matrix Partitioner)
presented in this article, compute optimal solutions for bipartitioning.  Another question that arises is about
the \NP-completeness~\cite{garey79} of the sparse matrix bipartitioning problem. It is known that the decision
problem of graph bipartitioning with a tolerated imbalance is \NP-complete~\cite[Theorem 3.1]{bui92} and so is
hypergraph partitioning~\cite[Chapter 6]{lengauer90},  but sparse matrix bipartitioning is a special case of
hypergraph bipartitioning (for instance, with vertices contained in only two hyperedges), and whether this problem
is \NP-Complete is still open.

The novelty of this paper is twofold: (i) we present an improvement of the previous state-of-the-art exact
algorithm~\cite{pelt15} by generalizing a matching-based lower bound on the necessary communication to a stronger
maximum flow-based bound, and by generalizing a packing bound (using ideas found by
Delling~\textit{et al.\ }\cite{delling14}); (ii) we formalize sparse matrix bipartitioning as a decision problem
and prove that it is \NP-complete.

The matrix bipartitioning problem that we solve by an exact algorithm can be formulated as follows. Given
an $m \times n$ sparse matrix with $|A|$ nonzeros and an allowed imbalance fraction of $\epsilon \geq 0$, find
disjoint subsets $A_1, A_2 \subseteq A$ such that
\begin{equation}
	A=A_1 \cup A_2,
\end{equation}
and
\begin{equation}
\label{eqn:imbal}
	|A_i| \leq (1+\epsilon ) \left\lceil \frac{|A|}{2} \right\rceil , ~\mathrm{for}~ i=1,2,
\end{equation}
and such that the communication volume $VOL(A_1,A_2)$ is minimal.

Here, the \emph{communication volume} is defined as the total number of rows and columns  that have nonzeros
in both subsets. Each of these \emph{cut} rows and columns gives rise to one communication of a single scalar
in a parallel SpMV.
\autoref{eqn:imbal} represents a constraint on the load balance of two processors of a parallel computer
executing the SpMV.

In this paper, we will only consider the communication volume as the metric to be minimized. Note that other
possible objectives, such as minimizing the maximum communication volume per processor or minimizing the total
number of messages, may also be relevant. For bipartitioning, however, these metrics need not specifically be
optimized: for two processors, the volume per processor is just half the total volume, if we count sending and
receiving data as equally important; furthermore, the total number of messages sent is at most two, so there is
not much to be optimized. Here, we ignore any costs needed for packing the data words into a message; these costs
are proportional to the number of data words to be sent.
For partitioning into a larger numbers of parts than two, recursive bipartitioning is
often used, and then the history of previous bipartitionings must be taken into account if for instance we also
want to minimize the total number of messages. Here, starting up a new message incurs an extra cost, above the
cost of sending a data word. This could be done by a sophisticated recursive hypergraph bipartitioning
approach~\cite{selvitopi17} that simultaneously reduces the communication volume and the number of messages,
thus solving a different optimization problem; this problem, however, is beyond the scope of the present paper.

Many exact partitioning algorithms have been developed for
graphs~\cite{delling14,karisch00,sensen01,felner05,hager13}. All these algorithms minimize the edge cut, not the
communication volume. Felner~\cite{felner05} solves a graph partitioning problem with uniform edge weights to
optimality with a purely combinatorial branch-and-bound method, reaching up to 100 vertices and 1000 edges.
Delling~\textit{et al.\ }\cite{delling14} solved larger problems using packing-tree bounds and graph
contractions, and solved  the open street map problem \allowbreak\texttt{luxembourg} with 114,599 vertices and
119,666 edges in less than a minute.

For hypergraphs, much less work has been done on exact partitioning~\cite{caldwell00,kucar04,bisseling05}.
Kucar~\cite{kucar04} uses integer linear programming (ILP)  to solve a problem with 1888 vertices, 1920 nets
(hyperedges), and 5471 pins (nonzeros) in three days of CPU time; the heuristic solver hMetis~\cite{karypis99b}
managed to find a solution in less than a second for the same problem, and it turned out to be optimal. Bisseling
and his team members~\cite{bisseling05} solved an industrial call-graph problem by formulating it as a hypergraph
partitioning problem with the cut-net metric, and they solved it heuristically by using Mondriaan and exactly by
an ILP method (in 9 days of CPU time).

For exact sparse matrix partitioning, the problem could in principle be formulated as a hypergraph bipartitioning
problem by using the fine-grain model~\cite{catalyurek01}: each matrix nonzero becomes a vertex in the hypergraph; 
the nonzeros in a row are connected by a row-net   and the nonzeros in a column by a column-net. Thus, we obtain
a hypergraph with $|A|$ vertices and $m+n$ nets, with the special property that each vertex is contained in
precisely two nets. One of these nets thus belongs to a group of $m$ pairwise disjoint row-nets,
and the other to a group of $n$ pairwise disjoint
column-nets. Furthermore, no two vertices have the same pair of nets. An exact general hypergraph partitioner
could then be used to solve the problem to optimality.  This, however, is less efficient than direct exact sparse
matrix partitioning, since the hypergraph partitioner would not exploit the aforementioned properties. In
contrast, the direct matrix approach imposes them by construction.

Previous work~\cite{pelt15} presented the first direct exact matrix partitioner, implemented in the
open-source software MondriaanOpt. This work was optimized and parallelized by Mumcuyan and
coworkers~\cite{mumcuyan18} who  reordered the matrix given to MondriaanOpt (thus changing the order in which the
search space was traversed), automatically choosing the best reordering method from a set of methods by a
machine-learning approach, and by parallelising the software for a shared-memory computer using OpenMP.  Our own
improvements, in the present article, are orthogonal to these extensions, so that they can be combined.

The remainder of this paper is organized as follows: Section~\ref{sec:np} presents the \NP-completeness proof 
for $\epsilon$-balanced matrix bipartitioning. Section~\ref{sec:opt} briefly reviews the previously mentioned
branch-and-bound algorithm~\cite{pelt15} that was implemented in MondriaanOpt, and presents the generalized bounds 
and their implementation. Section~\ref{sec:experiments} presents the experimental results, comparing MP to
MondriaanOpt for 233 small matrices, and giving results for 599 larger matrices that could not be solved by
MondriaanOpt within the allotted time. It also presents a comparison between two heuristic methods, PaToH and
Mondriaan, using these optimal partitionings.
Section~\ref{sec:concl} presents the conclusions and discusses possible future work.

		\section{Hardness results}
	\label{sec:np}
	In this section we will formally analyze matrix bipartitioning
	and prove that it is \NP-Complete, even if we fix the number 
	of processors to $k = 2$ (instead of leaving $k \geq 2$ arbitrary).
	This problem is a special case of the \NP-Complete hypergraph partitioning
	problem~\cite[Chapter 6]{lengauer90}.

	\subsection{Preliminaries}
	\label{1-preliminaries}
	To begin, let us define a formal decision-variant of the matrix
	partitioning problem for $k = 2$, based on the optimization
	variant described in \autoref{sec:intro} where the goal
	is to minimize the total communication volume. We formulate our
	decision problems in the style of Garey and Johnson~\cite{garey79}.
	Given a fixed $\epsilon \in [0, 1)$, we define \mbpt as follows:

	\begin{problem}
		\problemtitle{\mbpt}
		\probleminput{An $m \times n$ matrix $A$, whose nonzeros are precisely
			indexed by the set
			$Z \subseteq \{\, 1, \dots, m \,\} \times
				\{\, 1, \dots, n \,\}$, and an integer $M$, the required
			maximum volume.}
		\problemquestion{Does there exist a disjoint partitioning of $Z$
			into $Z_1 \cup Z_2$ such that $|Z_i| \leq (1+\epsilon)\tfrac{|Z|}{2}$
			and volume $VOL(Z_1, Z_2) \leq M$?}
	\end{problem}

	Note that $\epsilon$ is not part of the problem instance, but of the problem definition.
	Each $\epsilon$ induces its own partitioning problem, and we really have a class of
	bipartitioning problems here.

	Here $VOL(Z_1, Z_2)$ counts the number of rows and columns that have nonzeros
	in $Z_1$ and $Z_2$, as before.
	Additionally, to simplify presentation we will from here on, without loss of
	generality, assume that $(1+\epsilon)\tfrac{|Z|}{2}$ is an integer, rather
	than place rounding symbols everywhere. Ultimately, all the value of
	$\epsilon$ does is induce some integer upper bound on the $|Z_i|$.
	So for any $\epsilon$ we can, for the purposes of whichever problem instance
	we are considering, replace it with some $\epsilon'$ satisfying this integrality
	constraint, without affecting the bound on $|Z_i|$.
	Notice that requiring $|Z_i| \leq (1+\epsilon)\tfrac{|Z|}{2}$ is equivalent to
	requiring that $\big||Z_1| - |Z_2|\big| \leq \epsilon|Z|$, and that
	$\epsilon|Z|$ must be an integer as well.

	When thinking about the matrix bipartitioning problem, it is helpful
	to reformulate it in terms of graphs. Given an $m \times n$ matrix $A$
	we can define a bipartite adjacency graph $G(A) = (V(A), E(A))$ with
	$m$ vertices representing the rows of $A$, and $n$ vertices representing
	the columns, where a row vertex $r$ and a column vertex $c$ are connected
	if and only if $A_{rc}$ is nonzero.

	This equivalence extends to the partitioning problem. A bipartitioning of
	the nonzeros of $A$ corresponds to a bipartitioning of the edges of $G(A)$,
	and the rows and columns contributing to the final volume correspond
	precisely to the vertices with edges in both sides of the partition. See
	also \autoref{mat-equiv-fig}.

	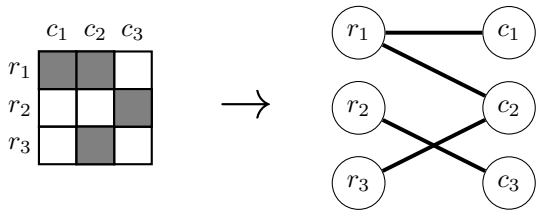
\begin{figure}[h]
		\begin{tikzpicture}
[
	box/.style={rectangle,draw=black,thick, minimum size=0.5cm},
]
		
	\begin{scope}
		\node at (2.5, 0.5){\huge$\rightarrow$};
	\end{scope}

	\begin{scope}[shift={(0, 0)}]
		\foreach \x in {0,0.5,1}{
			\foreach \y in {0,0.5,1}
				\node[box] at (\x,\y){};
		}
		\node[box,fill=gray] at (0, 1){};
		\node[box,fill=gray] at (0.5, 1){};
		\node[box,fill=gray] at (1, 0.5){};
		\node[box,fill=gray] at (0.5, 0){};
		\node[] at (-0.5,	0) {$r_3$};
		\node[] at (-0.5,	0.5) {$r_2$};
		\node[] at (-0.5,	1) {$r_1$};
		\node[] at (0,	1.5) {$c_1$};
		\node[] at (0.5,	1.5) {$c_2$};
		\node[] at (1,	1.5) {$c_3$};
	\end{scope}

	\begin{scope}[shift={(4, 0)}]
		\node[shape=circle,fill=white,draw=black,minimum size=8pt] (1)
			at (0, 1.5) {$r_1$};
		\node[shape=circle,fill=white,draw=black,minimum size=8pt] (2)
			at (0, 0.5) {$r_2$};
		\node[shape=circle,fill=white,draw=black,minimum size=8pt] (3)
			at (0, -0.5) {$r_3$};
		\node[shape=circle,fill=white,draw=black,minimum size=8pt] (4)
			at (2, 1.5) {$c_1$};
		\node[shape=circle,fill=white,draw=black,minimum size=8pt] (5)
			at (2, 0.5) {$c_2$};
		\node[shape=circle,fill=white,draw=black,minimum size=8pt] (6)
			at (2, -0.5) {$c_3$};

		\draw[line width=1.5pt] (4) -- (1) -- (5) -- (3);
		\draw[line width=1.5pt] (2) -- (6);
	\end{scope}
\end{tikzpicture}
		\centering
		\caption{Bipartite graph and matrix equivalence.}
		\label{mat-equiv-fig}
	\end{figure}
	
	This procedure is also reversible, i.e. for any
	bipartite graph $G$ on $m$ and $n$ vertices,
	we can construct a corresponding matrix $A$ of size $m \times n$ which has
	nonzeros precisely for the vertices that are connected in $G$. While the
	nonzero entries of this matrix can have any nonzero value, the associated nonzero
	pattern is uniquely determined by the edges of $G$.

	To this end, we define an equivalent bipartitioning problem on graphs that
	we will base our reduction on:

	\begin{problem}
		\problemtitle{\geb}
		\probleminput{Given a graph
			$G = (V, E)$ and an integer
			$M$.}
		\problemquestion{Does there exist a disjoint partitioning of $E$
			into $E_1 \cup E_2$ such that $|E_i| \leq (1+\epsilon)\tfrac{|E|}{2}$
			and
			$\big|\big(\bigcup_{e\in E_1} e \big) \cap
				\big(\bigcup_{e\in E_2} e \big)\big| \leq M$?\footnote{
			If we write an edge as the set $\{ u, v\} \subset V$, then
			$\bigcup_{e \in E_1} e$ gives precisely the set of vertices touched
			by $E_1$.}}
	\end{problem}

	We call a vertex with adjacent edges from both sides of the partitioning
	`cut'. The
	goal is to minimize the number of cut vertices. Additionally, when we
	explicitly need the partitioning/coloring of $E$, we will write it 
	as a map $\pi : E \to \{1,2\}$.

	\subsection{\geb is \NP-Complete}
	We will perform a reduction from the \gb problem, a classical problem
	first proven \NP-Complete
	for $\epsilon = 0$~\cite{gb-np-complete}, and later also for all
	$\epsilon\in[0,1)$~\cite{bui92,wagner1993between}.

	\begin{problem}
		\problemtitle{\gb}
		\probleminput{A graph $G = (V, E)$, an integer $M$.}
		\problemquestion{Does there exist a disjoint partitioning of $V$
			into $V_1 \cup V_2$ with $|V_i| \leq (1+\epsilon)\tfrac{|V|}{2}$
			such that
			$|\{\, \{u, v\} \in E\,\mid\, u \in V_1, v\in V_2 \,\}| \leq M$?}
	\end{problem}

	Analogously to the \geb problem, we call an edge with endpoints in both sides of the
	partitioning `cut'; the goal is then to minimize the number of
	cut edges. We similarly write a partitioning of $V$ as
	a map $\tau : V \to \{1, 2\}$. We can also think of $\tau$ as coloring
	the vertices in $V$, where one side of the partitioning has the color red, and one
	side has the color blue. This should not be confused, however, with the classical
	graph coloring problem, since we allow neighbouring vertices to have the
	same color.

	Let us give a sketch of our proof strategy:
	given an instance $(G, M)$ of the \gb problem, we
	will build a new graph $G'$, whose optimal solution under the 
	\geb problem will give us an optimal solution under \gb on $G$.

	We first define the construction of such an instance,
	giving the promised sketch directly after.
	Let us define the \textit{clique expansion} of a graph.
	A visual example (with smaller $S$, for legibility) is also given in
	\autoref{figure-graph-clique-conversion}.

	\begin{definition}
		\label{clique-expansion}
		Given a graph $G = (V, E)$,
		let $S = 4 + 2|V|\binom{|E|}{2}$.
		We define the clique expansion
		$\mathcal{K}(G) = (V', E')$ as first taking a disjoint union of $|V|$
		copies of the complete graph $K_S$. Then, labelling the
		edges in $E$ as $e_1, e_2, \dots, e_{|E|}$, for each edge
		$e_i = \{u, v\}$
		we merge the $i^{th}$ vertex of the clique $K_S$ representing
		$u$ with the $i^{th}$ vertex of the clique $K_S$ representing $v$.
	\end{definition}

	Here, the chosen clique size $S$ will allow us to prove several
	useful propositions later. Notice that by construction, each pair of
	cliques is merged at most once, each time in a previously unmerged vertex.
	As a consequence, while two cliques can share a single vertex,
	they do not share edges.

	Informally, for each vertex $u \in V$ we create a clique $K_u$ of size
	$S$. For every edge $\{u, v\} \in E$, we merge
	two vertices in the cliques $K_u$ and $K_v$ together into a single
	vertex. We then solve the \geb problem on the resulting graph $(V', E')$, and
	translate the resulting coloring of its edges into a coloring of the
	vertices of $(V, E)$. (Here, the colors correspond to the two parts in the
	partition.) We note that if each clique in $(V', E')$ is colored
	monochromatically then we can color each vertex in $(V, E)$ with
	the color of its corresponding clique in $(V', E')$. Then an edge
	between two differently colored vertices in $(V, E)$ will correspond
	precisely with a vertex shared by two differently colored cliques
	in $(V', E')$. The \geb problem gives us no guarantee that each clique
	is indeed colored monochromatically, but we will work around this later.

	\begin{figure}[h]
		% Fast clique drawing
\newcommand\single[3]{ % #1=labels, $2= n=number of nodes
	\foreach \x in {1,...,#2}{
		\pgfmathsetmacro{\ang}{360/#2}
		\pgfmathparse{(\x-1)*\ang}
		\node[shape=circle,fill=black,draw=black] (#1-\x)
			at (\pgfmathresult:1cm) {};
	}
	\foreach \x [count=\xi from 1] in {1,...,#2}{
		\foreach \y in {\x,...,#2}{
			\path (#1-\xi) edge[color=#3,line width=1.75pt] (#1-\y);
		}
	}
}
\begin{scaletikzpicturetowidth}{\linewidth}
\begin{tikzpicture}[scale=\tikzscale]

	% Arrows giving direction of conversion
	\begin{scope}
		\node at (4, 0){\huge$\rightarrow$};
	\end{scope}

	\begin{scope}
		\node at (8, -2.5){\huge$\downarrow$};
	\end{scope}

	\begin{scope}
		\node at (4, -5){\huge$\leftarrow$};
	\end{scope}

	% Initial graph
	\begin{scope}
		\node[shape=circle,fill=black,draw=black,minimum size=8pt] (1)
				at ($(-0.5,{sqrt(3)*-0.5})$) {};
		\node[shape=circle,fill=black,draw=black,minimum size=8pt] (2)
				at ($(-0.5,{sqrt(3)*0.5})$) {};
		\node[shape=circle,fill=black,draw=black,minimum size=8pt] (3)
				at (1, 0){};
		\node[shape=circle,fill=black,draw=black,minimum size=8pt] (4)
				at ($({1+sqrt(2)}, 0)$) {};

		\draw[line width=1.5pt] (1) -- (2) -- (3) -- (1);
		\draw[line width=1.5pt] (3) -- (4);
	\end{scope}

	% Uncolored cliques
	\begin{scope}[local bounding box=scope1,shift={(7, 0)}]
		\foreach \s[count=\si from 0] in {0,120,240}{
			\begin{scope}[shift={($(\s:1.1414)$)}, rotate=\s]
				\single{\si}{7}{black};
			\end{scope}
		}
		\begin{scope}[shift={(3.1415, 0)}, rotate=180]
			\single{4}{7}{black};
		\end{scope}
	\end{scope}

	% Colored cliques
	\begin{scope}[local bounding box=scope1,shift={(7, -5)}]
		\foreach \s[count=\si from 0] in {0,120,240}{
			\begin{scope}[shift={($(\s:1.1414)$)}, rotate=\s]
				\ifthenelse{\equal{\si}{0}}
					{\single{\si}{7}{red};}
					{\single{\si}{7}{blue};}
			\end{scope}
		}
		\begin{scope}[shift={(3.1415, 0)}, rotate=180]
			\single{4}{7}{red};
		\end{scope}
	\end{scope}

	% Final graph
	\begin{scope}[shift={(0, -5)}]
		\node[shape=circle,fill=blue,draw=blue,minimum size=8pt] (1)
				at ($(-0.5,{sqrt(3)*-0.5})$) {};
		\node[shape=circle,fill=blue,draw=blue,minimum size=8pt] (2)
				at ($(-0.5,{sqrt(3)*0.5})$) {};
		\node[shape=circle,fill=red,draw=red,minimum size=8pt] (3)
				at (1, 0){};
		\node[shape=circle,fill=red,draw=red,minimum size=8pt] (4)
				at ($({1+sqrt(2)}, 0)$) {};

		\draw[line width=1.5pt] (1) -- (2) -- (3) -- (1);
		\draw[line width=1.5pt] (3) -- (4);
	\end{scope}
\end{tikzpicture}
\end{scaletikzpicturetowidth}
		\centering
		\caption{Solving \gb using \geb.}
		\label{figure-graph-clique-conversion}
	\end{figure}

	Throughout this section we will make a slight abuse of
	terminology. A clique usually (and up until now) refers to any collection
	of pairwise connected vertices. However, from now on, when we talk about
	`cliques' in $\mathcal{K}(G)$ we will be referring specifically to the
	cliques corresponding to vertices, i.e. the cliques
	$\{\, K_u \,\mid\, u \in V \,\}$ in $\mathcal{K}(G)$.

	So after building $\mathcal{K}(G)$ from the \gb instance $G$,
	we can solve the \geb problem
	on it. We now want to show that both problems have optimal solutions
	of equal volume (cost).
	For a graph $G=(V, E)$, let $\sgb(G)$ denote the volume of any optimal
	solution of the \gb problem on $G$, and let $\sgeb(\mathcal{K}(G))$ denote
	the volume of any optimal solution of the \geb problem on its
	\textit{clique expansion} $\mathcal{K}(G)$.

	\begin{proposition}
		\label{geb-leq-gb}
		For any graph $G$ we have $\sgeb(\mathcal{K}(G)) \leq \sgb(G)$.
	\end{proposition}
	\begin{proof}
		Consider any valid $\epsilon$-balanced bipartitioning $\tau$ of
		$G = (V, E)$ of cost $M$ (that is, there are
		exactly $M$ edges $\{u, v\}$ with $\tau(u) \neq \tau(v)$).
		For any vertex $u$,
		color all edges in the corresponding clique $K_u$ in $\mathcal{K}(G)
			= (V', E')$ with the same color, i.e. for any edge $e$ in $K_u$
		let $\pi(e) := \tau(u)$, giving a partitioning $E' = E^{\prime}_1 \cup E^{\prime}_2$.

		Since each clique $K_u$ in $\mathcal{K}(G)$ has the same number of
		edges, $\binom{S}{2}$, and by
		assumption the partitioning of $V$ satisfies $|V_i| \leq (1+\epsilon)\tfrac{|V|}{2}$,
		it follows that $|E^{\prime}_i| \leq (1+\epsilon)\tfrac{|V|}{2}\binom{S}{2}
			= (1+\epsilon)\tfrac{|E'|}{2}$.

		Now let $s$ be a vertex in $\mathcal{K}(G)$. If $s$ is contained in
		only one clique, it cannot be cut,
		since we color the edges of
		each $K_u$ monochromatically. If $s$ is shared by two
		cliques $K_u$ and $K_v$, then $s$ corresponds to the edge
		$e = \{u, v\} \in E$, and
		we can see that this vertex is cut by $\pi$
		if and only if $e$ is cut by $\tau$ (since $K_u$ and
		$K_v$ are colored like $u$ and $v$ respectively).
		Since by construction,
		each vertex is in at most two cliques, there is no ambiguity.

		So the number of cut vertices in the induced partitioning $\pi$ of the
		edges of
		$\mathcal{K}(G)$ is exactly the number of cut edges in the original
		partitioning $\tau$ of the vertices of $G$. We can then minimize
		over all valid partitionings $\tau$ of $G$ to achieve the desired inequality.
	\end{proof}

	Unfortunately, the converse is harder to prove since we cannot guarantee
	that an optimal partitioning of $\mathcal{K}(G)$ colors each clique
	monochromatically. It turns
	out however, that we can still deterministically associate a color with
	each clique, provided we have an optimal partitioning.

	\begin{definition}
		Let $K$ be a clique and suppose we have a coloring of its edges. The
		dominating color of $K$ is a color $c$ such that there exists a vertex
		in $K$ with all of its adjacent edges colored $c$.
	\end{definition}

	While this property is not well-defined in general, it is for our
	restricted case:

	\begin{lemma}
		\label{dom-col}
		Given a graph $G = (V, E)$ and an optimal partitioning $\pi$
		of the edges of its clique
		expansion $\mathcal{K}(G)$. Then each clique $K_u$ in $\mathcal{K}(G)$
		has a well-defined dominating color.
	\end{lemma}
	\begin{proof}
		Fix a clique $K_u$ in $\mathcal{K}(G)$. We need to prove existence and
		uniqueness of its dominating color.

		First we prove uniqueness: to the contrary, assume there are two
		vertices $r, b$ in $K_u$ such that $r$ has only red edges adjacent,
		and $b$ only blue edges. Since $K_u$ is a clique, the edge $\{r, b\}$
		exists, which must be both red and blue, a contradiction.

		As for existence, assume to the contrary that every vertex in $K_u$
		has both blue and red edges adjacent. But this means each vertex in
		$K_u$ is cut by the partitioning $\pi$, and so the cost of
		this partitioning of $\mathcal{K}(G)$ is at least the clique size
		$S = 4+2|V|\binom{|E|}{2}$. One may verify that for any graph,
		$S > |E|$. However, $|E|$ is a trivial upper bound on the \gb problem
		on $(V, E)$ (in which we cut \textit{every} edge in $E$),
		which, by \autoref{geb-leq-gb} is an upper bound on the optimal
		partitioning of the edges of $\mathcal{K}(G)$. Since we assumed our
		partitioning $\pi$ is optimal, i.e. has cut size
		exactly equal to $\sgeb(\mathcal{K}(G))$, this implies that

		$$|E| \geq \sgb(G) \geq \sgeb(\mathcal{K}(G)) \geq S > |E|$$
		which is a contradiction, so there must exist a vertex that only has
		adjacent edges of a single color.
	\end{proof}

	In addition to the above, we would like to note in particular that by
	definition, if $K_u$ has dominating color $c$ (meaning there is a vertex
	$u'$ with all adjacent edges colored $c$), then any vertex in $K_u$
	has at least one adjacent edge with color $c$ (namely the one connecting
	it to $u'$).

	We now have the tools to formulate a proof strategy:
	we will color vertices in $G$ by
	the dominating color of their cliques in an optimal partitioning of
	$\mathcal{K}(G)$.

	\begin{proposition}
		\label{gb-leq-geb}
		For any graph $G$ we have $\sgb(G) \leq \allowbreak\sgeb(\mathcal{K}(G))$.
	\end{proposition}
	\begin{proof}
		Fix any optimal partitioning $\pi$ of $\mathcal{K}(G) = (V', E')$, and
		let $\tau$ color each vertex $u$ in $G = (V, E)$ with the dominating
		color of its associated clique $K_u$ in $\mathcal{K}(G)$. We would
		like to
		prove two things about this partitioning $\tau$:
		that the number of cut edges in $G$ is no more than the number of
		cut vertices in $\mathcal{K}(G)$, and that it is a balanced partitioning
		of $V$.
		
		We first show the number of edges cut by $\tau$ in $G$ is at most the
		number of vertices cut by $\pi$ in $\mathcal{K}(G)$. Suppose $\tau$
		cuts edge $e_i = \{u, v\} \in E$, that is, $\tau(u) \neq \tau(v)$.
		This means that the dominating colors of $K_u$ and $K_v$ are different,
		say without loss of generality that $K_u$ is red and $K_v$ is blue.
		Hence, the vertex $s$ in $\mathcal{K}(G)$ that corresponds to $e_i$,
		which we obtained during construction by merging the $i^{th}$
		vertex of $K_u$ with the $i^{th}$ vertex of $K_v$, must have
		red edges adjacent, because it is contained in $K_u$, and blue
		edges, because it is contained in $K_v$. So $\pi$ cuts $s$.
		Since for every edge $e_j \in E$ we merged different
		vertices (specifically, for $e_j$ we used the $j^{th}$ vertex of
		the two cliques), each edge in $E$ cut by $\tau$ has a unique
		corresponding vertex $s$ in $\mathcal{K}(G)$ cut by $\pi$,
		proving the first part.

		Next, to show that $\tau$ is a balanced partitioning of $V$, we will
		equivalently show that our optimal partitioning $\pi$ of $\mathcal{K}(G)$
		colors no more than $(1+\epsilon)\tfrac{|V|}{2}$ cliques with red as their
		dominating color, and the same for blue.

		Let $r,b \geq 0$, $r+b=|V|$, count these quantities, assuming
		without loss of generality that $r \geq b$. 
		First we derive a lower
		bound on the number of red edges in a clique in $\mathcal{K}(G)$.
		In each red clique we have at most
		$|E|$ cut vertices (since we assumed our partitioning was optimal,
		as in the proof of \autoref{dom-col}; in fact,
		across all cliques there are at most $|E|$ cut vertices, but for a
		lower bound this will suffice), and the edges between two such
		vertices may be blue, but none of the other $S - |E|$ vertices
		in this clique are
		cut, so all other edges should be red, and a lower
		bound on the number of red edges in $\pi$ is
		$$r \left( \binom{S}{2} - \binom{|E|}{2}\right).$$

		Similarly, we can find an upper bound for the number of blue edges
		by the following reasoning: we color each blue-dominated clique
		entirely blue, and as many edges as possible in each red-dominated
		clique (at most $\binom{|E|}{2}$, as before). This gives as
		an upper bound
		$$b\binom{S}{2} + r\binom{|E|}{2}.$$

		But since $\pi$ was an optimal \textit{balanced} partitioning of the
		edges of $\mathcal{K}(G)$, certainly the lower bound on the number
		of red edges must be smaller than or equal to the upper bound on
		the number of blue edges, plus the imbalance allowed by $\epsilon$:

		$$r\left( \binom{S}{2} - \binom{|E|}{2}\right)
			- \left(b\binom{S}{2} + r\binom{|E|}{2}\right)
			\leq \epsilon|E'| = \epsilon|V|\binom{S}{2}.$$

		Reordering terms gives:

		\begin{equation}
			\label{ineq-pf}
			(r-b-\epsilon|V|)\binom{S}{2} \leq 2r\binom{|E|}{2}.
		\end{equation}

		Recall that we took $S = 4 + 2|V|\binom{|E|}{2}$. Since $S \geq 4$
		we have $S \leq \binom{S}{2}$, as well as $r \leq |V|$. So if
		\autoref{ineq-pf} holds, then certainly the following holds:
		$$(r-b-\epsilon|V|)S \leq 2|V|\binom{|E|}{2}.$$

		Substituting $S$ and rewriting we get
		$$4(r-b-\epsilon|V|) + 2(r-b-1-\epsilon|V|)|V|\binom{|E|}{2} \leq 0.$$
 
		Since we assumed $\epsilon|V| \in \mathbb{N}$ we also have
		$(r-b-\epsilon|V|) \in\mathbb{Z}$. Then the inequality can only be true if
		$r-b-\epsilon|V| \leq 0$, and since we assumed $r \geq b$, this gives
		$$0\leq r-b\leq \epsilon|V|,$$
		hence our partitioning is balanced in $G=(V, E)$.

		Combining the obtained results, we can turn any optimal solution to the \geb
		problem on $\mathcal{K}(G)$ into a solution of equal value to the \gb
		problem on $G$, proving \autoref{gb-leq-geb}.
	\end{proof}

	We are now almost ready to show that \geb is \NP-Complete. All that remains is
	showing that it is in \NP:
	\begin{proposition}
		\label{poly-size}
		The size of $\mathcal{K}(G) = (V', E')$ is polynomial in the size of
		$G = (V, E)$.
	\end{proposition}
	\begin{proof}
		By construction, each vertex in $V$ induces a subgraph with
		$O(|V||E|^2)$
		vertices and $O(|V|^2|E|^4)$ edges. After merging, the graph will only
		become smaller. So $|V'|$ is $O(|V|^2 |E|^2)$ and $|E'|$ is
		$O(|V|^3|E|^4)$.
	\end{proof}

	We can now conclude:

	\begin{theorem}
		\geb is \NP-Complete.
	\end{theorem}
	\begin{proof}
		We claim $\gb \leq_{\mathcal{P}} \geb$.
		For a given instance of \gb $(G, M)$, by \autoref{geb-leq-gb} and
		\autoref{gb-leq-geb} we know
		$\sgb(G) \leq M$ if and only if $\sgeb(\mathcal{K}(G)) \leq M$. So
		if we can solve \geb on $\mathcal{K}(G)$ (which has size polynomial
		in the size of $G$, by \autoref{poly-size}) in polynomial time, we can
		also solve \gb on $G$ in polynomial time.
	\end{proof}

	\subsection{\mbpt is \NP-Complete}

	We now consider the original \mbpt problem. As mentioned in
	\autoref{1-preliminaries}, it is equivalent to partitioning the edges of
	a graph. We would like to immediately draw an equivalence between \mbpt
	and \geb -- but note that for a matrix $A$ the associated graph $G(A)$ is
	always bipartite. Therefore, not every graph on which we might want to solve the
	\geb problem can be immediately converted to an equivalent matrix in the
	context of the \mbpt problem. However, we can resolve this:

	\begin{definition}
		Given a graph $G = (V, E)$, its \textit{edge-split graph}
		$\mathcal{S}(G) = (V', E')$ is given as:

		$$V' = V \cup \{\, v_e \,\mid\, e \in E \,\}$$
		$$E' = \bigcup_{e = \{u, w\} \in E} \{\,\{u, v_e\},\, \{v_e, w\}\,\}.$$
	\end{definition}

	In other words, we replace each edge by a path of length two. The resulting
	graph is bipartite (with sides $V$ and $V' \setminus V$). Using
	this bipartite extension of a graph, we can build a matrix and use \mbpt
	to solve the \geb problem.

	First, we prove that we can safely take the \textit{edge-split graph}
	without affecting the \geb problem.

	\begin{proposition}
		\label{split-graph}
		For any graph $G$, we have $\sgeb(G) = \allowbreak\sgeb(\mathcal{S}(G))$.
	\end{proposition}

	\begin{proof}
		Let $G = (V, E)$ and $\mathcal{S}(G) = (V', E')$.
		\begin{enumerate}
			\item[$(\geq)$] Let $\pi$ be an optimal coloring of $E$. We define
					a coloring $\pi'$ of $E'$ as follows. Let
					$e' = \{u, v_e\} \in E'$ with $u \in V \subseteq V'$ and
					$v_e \in V' \setminus V$. So $e'$ is half of the length-two
					path associated with $e$ in $G$. We set
					$\pi'(e') = \pi(e)$, that is, we give each edge in $E'$ the
					color of the edge in $E$ that induced it.

					Then no vertices in $V'\setminus V$ were cut, since
					both edges in a path have the same color, and the vertices
					cut by $\pi'$ in $V \subseteq V'$ are precisely those
					cut by $\pi$ in $V$.
			\item[$(\leq)$] Let $\pi'$ be any optimal coloring of $E'$. Now for
					each edge $e \in E$ there are three possibilities:
					\begin{itemize}
						\item Both corresponding edges in $E'$ are colored red
							by $\pi'$.
						\item Both corresponding edges in $E'$ are colored blue
							by $\pi'$.
						\item The corresponding edges in $E'$ are colored with
							two colors.
					\end{itemize}
					Let these quantities be counted by $n_r, n_b$ and $n_{rb}$,
					respectively. Then $|E| = n_r + n_b + n_{rb}$. Also the number
					of red and blue edges in $\pi'$ is counted by $2n_r + n_{rb}$
					and $2n_b + n_{rb}$, respectively. But by the
					balancing constraint we get
					$$|(2n_r + n_{rb}) - (2n_b + n_{rb})| \leq \epsilon|E'| = 2\epsilon|E|.$$

					From this we can derive that
					$|n_r - n_b| \leq \epsilon|E|$. That is, if we have an edge
					in $G$ whose associated pair in $\mathcal{S}(G)$ is monochrome,
					then copying these colors back into a coloring (for now ignoring
					those edges in $G$ whose associated pair is not monochrome)
					gives a coloring that already satisfies the load balancing
					constraint. This induced partial coloring will form the basis of
					our optimal coloring of $E$.
					
					What remains is to assign the
					bicolored pairs a color. We claim that we can recolor such
					a pair to a single color without increasing the volume of
					the solution. Indeed suppose without loss of generality
					that $e = \{u, w\} \in E$ satisfies $\pi'(\{u, v_e\}) = 1$
					and $\pi'(\{v_e, w\}) = 2$. Now construct a coloring
					$\pi^{\prime\prime}$ identical to $\pi'$ except that
					$\pi^{\prime\prime}(\{v_e, w\}) = 1$.

					So $\pi'$ and $\pi^{\prime\prime}$ are identical on the
					edges adjacent to $V'\setminus\{v_e, w\}$, and hence cut
					the same vertices here. Note that $\pi'$ also cuts $v_e$,
					but $\pi^{\prime\prime}$ does not. So even if $w$ now
					goes from `not cut' to `cut', this does not matter since
					`uncutting' $v_e$ compensates for this.

					Applying a change like this decreases $n_{rb}$ by one, and
					increases $n_r$ or $n_b$ by one (depending on the choice of
					color), therefore changing $|n_r - n_b|$ by one.
					Let us now distinguish two cases:

					\begin{description}
						\item[($\epsilon|E| \geq 1$)] Since a priori we have
							$|n_r - n_b| \leq \epsilon|E|$, we can
							repeatedly either increase $n_r$ or $n_b$
							by one while never violating the balancing
							constraint.
						\item[($\epsilon|E| = 0$)] In this case we must in
							fact have $$2n_r + n_{rb} = 2n_b + n_{rb}$$
							and hence $n_r = n_b$. Since $|E|$ is even
							(otherwise no partitioning exists), we get
							that $n_{rb} = |E| - (n_r + n_b)$ is even
							as well. Therefore we can pair up such
							bicolored pairs, making one of them fully red
							and one fully blue at each step, causing no
							change in the value of $|n_r - n_b|$.
					\end{description}

					In either situation we can find a new coloring $\pi^{\prime\prime}$
					of $E'$ with equal volume but with no bicolored pairs.
					Additionally, the number of monochrome pairs of each color
					satisfies the load balancing constraint
					$|n_r - n_b| \leq \epsilon|E|$ of the original graph. So we
					can then map this coloring to $E$, like in case $(\geq)$.
		\end{enumerate}
	\end{proof}

	\begin{figure*}[hbt!]
		\begin{tikzpicture}[box/.style={rectangle,draw=black,thick, minimum size=0.5cm}]
		
	\begin{scope}
		\node at (3.5, 0){\huge$\rightarrow$};
	\end{scope}

	\begin{scope}
		\node at (9, 0){\huge$\rightarrow$};
	\end{scope}

	\begin{scope}
		\node at (11.3, -2){\huge$\downarrow$};
	\end{scope}
			
	\begin{scope}
		\node at (3.5, -4.2){\huge$\leftarrow$};
	\end{scope}

	\begin{scope}
		\node at (9, -4.2){\huge$\leftarrow$};
	\end{scope}
		
	\begin{scope}
		\node[shape=circle,fill=white,draw=black,minimum size=8pt] (1) at ($(-0.5,{sqrt(3)*-0.5})$) {$v_1$};
		\node[shape=circle,fill=white,draw=black,minimum size=8pt] (2) at ($(-0.5,{sqrt(3)*0.5})$) {$v_2$};
		\node[shape=circle,fill=white,draw=black,minimum size=8pt] (3) at (1, 0){$v_3$};
		\node[shape=circle,fill=white,draw=black,minimum size=8pt] (4) at ($({1+sqrt(2)}, 0)$) {$v_4$};
		
		\draw[line width=1.5pt] (1) -- (2) -- (3) -- (1);
		\draw[line width=1.5pt] (3) -- (4);
	\end{scope}

	\begin{scope}[local bounding box=scope1,shift={(5, 0)}]
		\node[shape=circle,fill=white,draw=black,minimum size=8pt] (1) at ($(-0.5,{sqrt(3)*-0.5})$) {$v_1$};
		\node[shape=circle,fill=white,draw=black,minimum size=8pt] (5) at ($(-0.5,0)$) {$e_1$};
		\node[shape=circle,fill=white,draw=black,minimum size=8pt] (2) at ($(-0.5,{sqrt(3)*0.5})$) {$v_2$};
		\node[shape=circle,fill=white,draw=black,minimum size=8pt] (6) at ($(0.25,{sqrt(3)*0.25})$){$e_2$};
		\node[shape=circle,fill=white,draw=black,minimum size=8pt] (3) at (1, 0){$v_3$};
		\node[shape=circle,fill=white,draw=black,minimum size=8pt] (7) at ($(0.25,{sqrt(3)*-0.25})$){$e_3$};
		\node[shape=circle,fill=white,draw=black,minimum size=8pt] (4) at ($({1+sqrt(2)+0.4}, 0)$) {$v_4$};
		\node[shape=circle,fill=white,draw=black,minimum size=8pt] (8) at ($({1+0.5*sqrt(2)+0.2}, 0)$){$e_4$};
			
		\draw[line width=1.5pt] (1) -- (5) -- (2) -- (6) -- (3) -- (7) -- (1);
		\draw[line width=1.5pt] (3) -- (8) -- (4);
	\end{scope}

	\begin{scope}[local bounding box=scope1,shift={(5, -4.2)}]
		\node[shape=circle,fill=white,draw=black,minimum size=8pt] (1) at ($(-0.5,{sqrt(3)*-0.5})$) {$v_1$};
		\node[shape=circle,fill=white,draw=black,minimum size=8pt] (5) at ($(-0.5,0)$) {$e_1$};
		\node[shape=circle,fill=white,draw=black,minimum size=8pt] (2) at ($(-0.5,{sqrt(3)*0.5})$) {$v_2$};
		\node[shape=circle,fill=white,draw=black,minimum size=8pt] (6) at ($(0.25,{sqrt(3)*0.25})$){$e_2$};
		\node[shape=circle,fill=white,draw=black,minimum size=8pt] (3) at (1, 0){$v_3$};
		\node[shape=circle,fill=white,draw=black,minimum size=8pt] (7) at ($(0.25,{sqrt(3)*-0.25})$){$e_3$};
		\node[shape=circle,fill=white,draw=black,minimum size=8pt] (4) at ($({1+sqrt(2)+0.4}, 0)$) {$v_4$};
		\node[shape=circle,fill=white,draw=black,minimum size=8pt] (8) at ($({1+0.5*sqrt(2)+0.2}, 0)$){$e_4$};
			
		\draw[line width=3pt,draw=red] (1) -- (5) -- (2) -- (6) -- (3);
		\draw[line width=3pt,draw=blue] (1) -- (7) -- (3) -- (8) -- (4);
	\end{scope}
	\begin{scope}[shift={(0, -4.2)}]
		\node[shape=circle,fill=white,draw=black,minimum size=8pt] (1) at ($(-0.5,{sqrt(3)*-0.5})$) {$v_1$};
		\node[shape=circle,fill=white,draw=black,minimum size=8pt] (2) at ($(-0.5,{sqrt(3)*0.5})$) {$v_2$};
		\node[shape=circle,fill=white,draw=black,minimum size=8pt] (3) at (1, 0){$v_3$};
		\node[shape=circle,fill=white,draw=black,minimum size=8pt] (4) at ($({1+sqrt(2)}, 0)$) {$v_4$};
		
		\draw[line width=3pt,draw=red] (1) -- (2) -- (3);
		\draw[line width=3pt,draw=blue] (1) -- (3) -- (4);
	\end{scope}
	\begin{scope}[shift={(10.5,-0.7)}]
		\foreach \x in {0,0.5,1,1.5}{
			\foreach \y in {0,0.5,1,1.5}
				\node[box] at (\x,\y){};
		}
		\node[box,fill=gray] at (0	,1.5){};
		\node[box,fill=gray] at (1	,1.5){};
		\node[box,fill=gray] at (0	,1){};
		\node[box,fill=gray] at (0.5	,1){};
		\node[box,fill=gray] at (0.5	,0.5){};
		\node[box,fill=gray] at (1	,0.5){};
		\node[box,fill=gray] at (1.5	,0.5){};
		\node[box,fill=gray] at (1.5	,0){};
		\node[] at (-0.5,	0) {$v_4$};
		\node[] at (-0.5,	0.5) {$v_3$};
		\node[] at (-0.5,	1) {$v_2$};
		\node[] at (-0.5,	1.5) {$v_1$};
		\node[] at (0,	2) {$e_1$};
		\node[] at (0.5,	2) {$e_2$};
		\node[] at (1,	2) {$e_3$};
		\node[] at (1.5,	2) {$e_4$};
	\end{scope}
	\begin{scope}[shift={(10.5,-5)}]
		\foreach \x in {0,0.5,1,1.5}{
			\foreach \y in {0,0.5,1,1.5}
				\node[box] at (\x,\y){};
		}
		\node[box,fill=red] at (0	,1.5){};
		\node[box,fill=blue] at (1	,1.5){};
		\node[box,fill=red] at (0	,1){};
		\node[box,fill=red] at (0.5	,1){};
		\node[box,fill=red] at (0.5	,0.5){};
		\node[box,fill=blue] at (1	,0.5){};
		\node[box,fill=blue] at (1.5	,0.5){};
		\node[box,fill=blue] at (1.5	,0){};
		\node[] at (-0.5,	0) {$v_4$};
		\node[] at (-0.5,	0.5) {$v_3$};
		\node[] at (-0.5,	1) {$v_2$};
		\node[] at (-0.5,	1.5) {$v_1$};
		\node[] at (0,	2) {$e_1$};
		\node[] at (0.5,	2) {$e_2$};
		\node[] at (1,	2) {$e_3$};
		\node[] at (1.5,	2) {$e_4$};
	\end{scope}
\end{tikzpicture}
		\centering
		\caption{Solving \geb using \mbpt.}
		\label{figmgeb}
	\end{figure*}

	Now we can turn any graph into a bipartite graph without changing its
	smallest edge bisection. Using this we can prove the main theorem of this
	section:

	\begin{theorem}
		\label{main-thm}
		\mbpt is \NP-Complete.
	\end{theorem}
	\begin{proof}
		We will show that $$\geb \leq_{\mathcal{P}} \mbpt$$ Given a graph $G$,
		let $G'=\mathcal{S}(G)=(L\cup R, E')$. Create a $|L| \times |R|$ matrix
		$A$ with $A_{ij} = 1$ if $i \in L$ and $j \in R$ are connected by an edge
		from $E'$, and $0$ otherwise.

		We can now solve the \mbpt problem on $A$, and using the correspondence
		between matrices and bipartite graphs described in
		\autoref{1-preliminaries}, we can turn this into an $\epsilon$-balanced
		partitioning of
		$E'$, since we have a correspondence between the partitioning of edges
		and the partitioning of nonzeros, and a correspondence between cutting
		vertices and cutting rows and columns. This is displayed in \autoref{figmgeb}.

		Now \autoref{split-graph} and its proof give us a constructive algorithm to
		transition between $G'$ and $G$, solving \geb on $G$.
	\end{proof}

	\section{Exact Algorithm}
\label{sec:opt}

In this section we give an exact algorithm for finding an optimal
$\epsilon$-balanced bipartitioning of a matrix, extending the
branch-and-bound
algorithm by Pelt and Bisseling \cite{pelt15}. They approach the
problem with a branch-and-bound technique and use combinatorial
bounds to bound the volume of a partial partitioning from below.
A major drawback of their bounds however, is that they are in some
sense `local'. They only consider the direct neighbourhoods of
assigned sections of the matrix (or rather, its underlying graph).
We extend both bounds to the entire graph to take
full advantages of its connectivity. We will first briefly discuss
the branch-and-bound algorithm, before looking at the two classes
of bounds.

\subsection{Branch and Bound}
Recall that a branch-and-bound algorithm initially starts with
(a representation of) the whole solution space, and then repeatedly
branches on properties of the solutions until these are refined
enough that they specify a single solution (this is a leaf in the
branch-and-bound tree). When the properties are chosen carefully, we
may prune (`bound') large parts of the search tree well before we reach
a leaf.

In the case of matrix bipartitioning, a first obvious choice would 
be to branch on which partition to put each nonzero in. For an $m \times n$
matrix with $N$ nonzeros this results in $2^N$ leaf nodes. However,
this is not our only option.
Instead, we can branch on the status of each of the rows and columns of
the matrix: each of them is either entirely red, entirely blue, or `cut',
i.e. it contains both colors. As a result, we only have $3^{m+n}$ leaves,
which is already smaller than $2^N$ when $m + n < \log_3(2) N \approx 0.63N$.
In fact, not all of the $3^{m+n}$ states are even reachable: if a row
and column intersect in a nonzero, we cannot mark one of them as
\textit{red} and one of them as \textit{blue}
(i.e., we do require assignments to be consistent).

When we traverse the branch-and-bound tree, at each stage we have a
`partial assignment', where some of the rows and columns are
\textit{red}, \textit{blue}, or \textit{cut}, and some are still
unassigned. For a given matrix $A$ and its bipartite graph representation
$G(A) = (V, E)$ (recall the equivalence from \autoref{sec:np}),
we will write $R \subseteq V$ (resp. $B, C \subseteq V$)
for the vertices (corresponding to rows and columns) that were assigned
\textit{red} (resp. \textit{blue}, \textit{cut}). Additionally, while
all remaining vertices are unassigned, they may still be connected to
vertices in $R, B$ and $C$. For example, if an unassigned column
vertex $u$ is adjacent to a row vertex $r \in R$, this means that $A_{ru}$
is nonzero. In particular, since row $r$ is red, we cannot make $u$
blue. So we will call $u$ \textit{partially red}, with the corresponding
subset of $V$ written as $P_R$ (with $P_B$ defined analogously).
Finally, an unassigned vertex may have neighbours in both $R$ and $B$.
Although it is unassigned, we are forced to cut this vertex.
Because of this, we assume that whenever such a vertex is created, we
immediately cut it. Therefore we will ignore these vertices.

Now we need to find a lower bound for any extension to a partial partition. First
we can take $|C|$ as our starting point. We then turn to the unassigned
region of the bipartite graph (matrix) to derive stronger bounds.

\subsection{Flow Bounds}

\begin{figure*}[htb!]
	\centering
	\begin{tikzpicture}[box/.style={rectangle,draw=black,thick, minimum size=1cm}]

\begin{scope}[shift={(-7, 1)}]
\foreach \x in {0,1,2}{
	\foreach \y in {0,1,2}
		\node[box] at (\x,\y){};
}
\node[box,fill=red] at (0,0){};
\node[box,fill=lightgray] at (0,2){};
\node[box,fill=blue] at (1,1){};
\node[box,fill=lightgray] at (2,1){};
\node[box,fill=lightgray] at (2,2){};
\node[] at (-1, 0) {$0$};
\node[] at (-1, 1) {$-$};
\node[] at (-1, 2) {$-$};
\node[] at (0, 3) {$-$};
\node[] at (1, 3) {$1$};
\node[] at (2, 3) {$-$};
\end{scope}

\node[shape=circle,fill=red  ,draw=red  ,minimum size=13pt] (1) at (0, 0) {};
\node[shape=circle,fill=black,draw=black,minimum size=13pt] (4) at (3, 0) {};
\node[shape=circle,fill=black,draw=black,minimum size=13pt] (2) at (0, 2) {};
\node[shape=circle,fill=blue ,draw=blue ,minimum size=13pt] (5) at (3, 2) {};
\node[shape=circle,fill=black,draw=black,minimum size=13pt] (3) at (0, 4) {};
\node[shape=circle,fill=black,draw=black,minimum size=13pt] (6) at (3, 4) {};

\node[] at(-1, 0) {Row $3$};
\node[] at(-1, 2) {Row $2$};
\node[] at(-1, 4) {Row $1$};
\node[] at( 4.3, 0) {Column $3$};
\node[] at( 4.3, 2) {Column $2$};
\node[] at( 4.3, 4) {Column $1$};

\draw[line width=2pt] (1) -- (6) -- (3) -- (4) -- (2) -- (5);

\end{tikzpicture}
	\caption{Looking at the matrix, it is not a priori clear that row 3
		and column 2 are connected. In the graph representation of the
		matrix we can clearly see the connecting path.}
	\label{fig-flow}
\end{figure*}
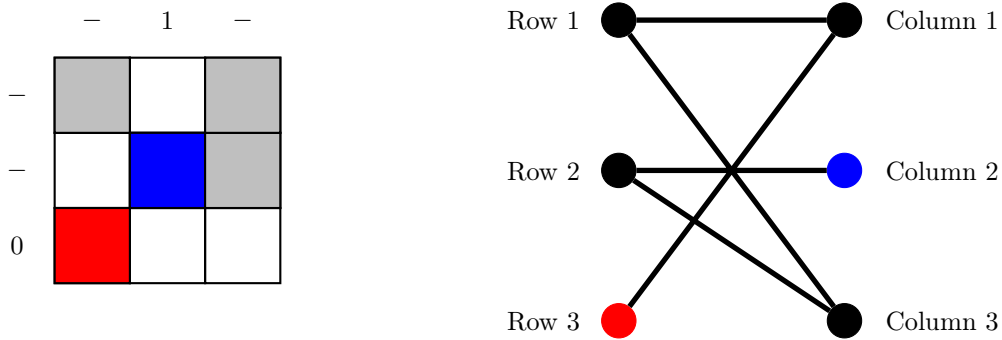

The first way to find a lower bound on a given partial partitioning
is to consider
the connectivity of the underlying graph. The original
algorithm~\cite{pelt15} uses the following matching bound:
let us consider some $p \in P_R$ and
$q \in P_B$ with $\{p, q\} \in E$. We note that $p$ has a red
edge adjacent (through its adjacent vertex in $R$), and $q$ has
a blue edge adjacent. Clearly no matter what color we give
$\{p, q\}$, we will have to cut one of $p$ and $q$.

We can improve upon this by finding a maximal set of such edges that
are vertex-disjoint (this is necessary, since otherwise we could
resolve two edges $\{p, q\}$ and $\{p, q'\}$ just by cutting one
vertex $p$).
The relevant graph is just our bipartite graph
$G(A) = (V, E)$ restricted to
$P_R \cup P_B \subseteq V$, keeping only edges with one
endpoint in $P_R$ and one in $P_B$. Bipartite
matching is a classical problem we can solve in polynomial
time. Then the size of the maximum matching is a lower bound
on the number of vertices that still have to be cut by any
extension of the current partial assignment.

This bound does not take full advantage of the connectivity
of the graph since it only considers single edges that are
adjacent to both a red and a blue edge (through $P_R$ and
$P_B$). There is no need to restrict ourselves to single edges
however, and this is especially clear when we consider the graph
formulation of the problem. We now introduce a new flow bound.

Consider a path from $P_R$ to $P_B$ avoiding $R \cup B \cup C$, that
is, a sequence of vertices $v_1, v_2, \dots, v_k$ with
$v_i \not\in R\cup B\cup C$ for $1 \leq i\leq k$, and in particular
$v_1 \in P_R$ and $v_k \in P_B$, such that $\{v_i, v_{i+1}\} \in E$
for all $1 \leq i < k$. This corresponds to a series of intersecting
rows and columns, the first of which contains a red nonzero
(corresponding to the edge between $v_1$ and its neighbour in $R$)
and the last a blue nonzero. One can see this in \autoref{fig-flow}.
But then any extension of the current partial assignment that colors
(assigns) all remaining vertices, must cut one
of these vertices. If $v_1$ is not cut then it must be fully red, making
$v_2$ partially red, and so on.

So if a single path between $P_R$ and $P_B$ implies that we have to cut at
least one vertex, how do we extend this to multiple paths? Here we run in
to the same issue as with the matching bound: if we have two paths that
share a vertex, we can just cut that vertex to separate red and blue edges,
for a cost of $1$.

Hence, to prove a lower bound of more than $1$ we have to require that the
paths are disjoint. In particular, the best lower bound we can hope for using
this technique
will be the maximum number of vertex-disjoint paths between $P_R$ and $P_B$.
Note that these paths must be vertex-disjoint in $P_R$ and $P_B$ as well,
since those vertices may also still be cut. Alternatively, when we imagine
$R$ and $B$ as a single vertex, we are looking for a maximum number of paths
from $R$ to $B$ that are vertex-disjoint outside of $R$ and $B$. In fact,
a theorem by Menger shows that this is actually exactly the size of the
smallest vertex cut.

\begin{theorem}{(Menger \cite{menger1927allgemeinen})}
	Let $G$ be a finite undirected graph, and let $u$ and $v$ be two
	non-adjacent vertices in $G$, then the size of the smallest vertex cut
	separating $u$ and $v$ equals the maximum number of vertex-disjoint
	paths between $u$ and $v$.
\end{theorem}

Unfortunately, the resulting bipartitioning may be very unbalanced. Intuitively,
the flow bound might be small if there is a chokepoint between $R$ and $B$.
However, the actual optimal bipartitioning might be much larger if this chokepoint
is biased towards one of $R$, $B$. We would like to correct for this by adding
another bound which considers the sizes of the neighbourhoods of $R$ and $B$
rather than their connectivity. This motivates the next class of bounds we
introduce.

\subsection{Packing Bounds}
The second bound used by Pelt and Bisseling~\cite{pelt15} is a local packing bound.
Let $E(R)$ denote all edges that are colored red by the current partial assignment
(that is, all edges adjacent to a vertex in $R$). For each partially red vertex
$p \in P_R$, let $$N_{\text{free}}(p) = \{\, e\in E \,|\,
	\text{$e$ is adjacent to $p$, $e\not\in E(R)$}\}$$ denote all its adjacent
unassigned edges. Note that all these edges are adjacent to a red edge (through
$p$). If we do not want to cut any more vertices, we have to color all of
$N_{\text{free}}(p)$ red. But if $|E(R)| + \sum_{p \in P_R} |N_{\text{free}}(p)|$
is greater than $(1+\epsilon)\tfrac{|E|}{2}$ this will lead to an unbalanced
partition, and so we are forced to cut some of the vertices $p\in P_R$. Since we are
looking for a lower bound, we can greedily take those vertices with largest
$|N_{\text{free}}(p)|$ until the sum is small enough again.

Note that we implicitly assume all $N_{\text{free}}(p)$ are disjoint, so that we
can assign their edges independently. Since $G(A) = (V, E)$ is bipartite, this is
true if we consider each side of the bipartition separately (that is, the rows and
the columns). We can then do the same for $P_B$, and add all unavoidable cuts together
to get the packing bound.

As with the matching bound this bound only considers the direct neighbourhood of 
$R \cup B$ and leaves large sections of the graph unexamined. Here too, we use the
fact that a path between a red edge and blue edge must have a cut somewhere in between
to introduce a new, stronger packing bound.
In particular we will look at whole subgraphs adjacent to $P_R$, rather than only
at free edges incident to vertices in $P_R$. This bound is based
on a similar approach taken by Delling and coworkers~\cite{delling14}.

\begin{definition}
	Given a graph $(V, E)$ with a partial assignment $R, B, C \subseteq V$,
	then an \emph{$R$-adjacent subgraph} $(V', E')$ is a tuple of subsets
	$V' \subseteq V$, $E' \subseteq E$, satisfying the following
	properties:
	\begin{itemize}
		\item[(1)] $V'$ is disjoint from $R \cup B \cup C$.
		\item[(2)] For any distinct $e_1, e_2 \in E'$ such that $u \in V$
			is an endpoint of both $e_1$ and $e_2$, we have $u \in V'$.
		\item[(3)] $(V', E')$ is path-connected with respect to edges, i.e.
			for any $e_1, e_2 \in E'$ we can find $f_1, \dots, f_k \in E'$
			pairwise incident, with $e_1 = f_1$ and $f_k = e_2$.
		\item[(4)] $(V', E')$ is adjacent to $R$ ($V' \cap P_R \neq \emptyset$).
		\item[(5)] All edges in $E'$ are free in the partial partitioning
			$R, B, C$ (that is, no edges are adjacent to $R$ or $B$).
	\end{itemize}
\end{definition}

We can now use these subgraphs to find a lower bound on any extension of our partial
assignment. Indeed, notice that for any edge in $E'$ we can find a path (property 3)
to an edge adjacent to $R$ (property 4) with all internal vertices in $V'$ (property 2)
and unassigned (property 1). Therefore coloring at least one of these edges blue
requires us to cut at least one vertex in $V'$.

Note that the definition intentionally does not require us to add both endpoints
of an edge in $E'$ to our vertex set $V'$. This is because if we have an edge
between our $R$-adjacent subgraph and $C$, this edge is still unassigned and we
can therefore still find a path containing an internal vertex that has to be cut,
in the event that this edge is colored blue.

We can now formulate a strategy to find a new, \textit{extended packing bound}:
we find a maximal collection of $R$-adjacent
subgraphs $(V_1, E_1), (V_2, E_2), \dots (V_k, E_k)$ that are pairwise disjoint.
Denoting by $E(R)$ all edges incident to $R$, and noticing that all edges in $E_i$
are unassigned by the current partial assignment $R, B, C$ (property 5),
then if we do not want to cut any more vertices and color all of the $R$-adjacent
subgraphs red, this results in $|E(R)| + \sum_{i=1}^k |E_i|$ red edges. If this is
larger than $(1+\epsilon)\tfrac{|E|}{2}$, any resulting partitioning would be
unbalanced, and so we must cut some of the subgraphs. To find a lower bound we can
again assume the ideal case where we cut the largest subgraphs (in terms of $|E_i|$)
first, at a cost of one cut per subgraph. We can compute a similar quantity
for $B$-adjacent subgraphs and add the results together for the extended packing bound.

\subsection{Implementation notes}
While the previous sections describe how to lower bound the volume given the
\textit{existence} of these specific paths and subgraphs, we are also interested in
\textit{finding} them. Additionally, there are various decisions one needs to make
in the implementation process of a branch-and-bound algorithm, which we detail below
in the interest of reproducibility.

\subsubsection{Branching strategy}
While we already specified that we branch on marking a row or column as
\textit{red}, \textit{blue}, or \textit{cut}, the order in which we
select the rows and columns for branching could significantly affect the
performance of the algorithm as this selection prescribes the order in which the
entire search space is traversed.

Intuitively, it makes sense to branch on rows and columns with more free
nonzeros, since their assignment affects the balance of the bipartitioning
the most, and their high connectivity suggests they may be useful as
sources of paths in the Flow Bound or subgraphs in the Extended Packing
Bound. Thus, at each step we select for branching a row or column $u$
with a maximal number of unassigned nonzeros, breaking ties arbitrarily.
Additionally, since the goal is to cut as few rows and columns as possible,
we traverse the `\textit{cut}' subtree last, and since the goal is to balance
the bipartition, we traverse first the subtree that assigns $u$ to the smallest
side in the bipartition. We note that the previous program, MondriaanOpt, employs the
same ordering strategy. As a remark, we refer to the paper by
Mumcuyan et. al.~\cite{mumcuyan18} who show other branching strategies
can be faster, and learn to predict the optimal strategy based on matrix
statistics. It turns out that the strategy that performs best on average
is exactly the one we implemented (branching on `\textit{cut}' vertices
last and ordering by number of nonzeros).

\subsubsection{Initial upper bound}
To correctly prune, our branch-and-bound algorithm needs an upper bound to
compare its lower bound against. Before we have found our first feasible
solution we could use the trivial $\min(m,n) + 1$ upper bound. Although
this upper bound is in some sense tight (consider an odd square matrix
with only one zero), for sparse matrices it is usually quite bad and
forces our algorithm to consider many suboptimal solutions before arriving
at better ones.

Instead, we would like to run the algorithm with an upper bound as tight
as possible. Hence, we use a technique common in branch-and-bound algorithms
where we run our algorithm with an initial (strict) upper bound of $U_1 = 1$,
and rerun with $U_{i+1} = \lceil \frac{5}{4} U_i \rceil$ until we have found
a solution. One may interpret this search technique as a variation of
\textit{Iterative Deepening A*}, considering volume instead of path length.

\subsubsection{Implementing the flow bound}
Finding a maximal set of vertex-disjoint paths is a classical maximum flow
problem. We can solve it by duplicating each vertex, and connecting the resulting
vertex pairs
with a capacity $1$ edge to enforce that every vertex be used only once.
A thorough exposition of this may be found in
the book by Cormen et al.~\cite[chapter 26]{thomas2001introduction}.

To leave as many edges for the packing bound as possible, we compute the
maximum flow using shortest paths~\cite{dinic1970algorithm}. Finally, while
the maximum flow problem may be solved in polynomial time, it still requires
computation over the entire graph, which may slow down our algorithm for large
matrices. Instead, we reuse the flow from the previous step in the
branch-and-bound algorithm. Recomputing the flow then involves finding
augmenting paths involving the modified vertex, which speeds up computation
considerably especially when this subgraph is small, lower in the search tree.

\subsubsection{Extended Packing Bound}
The quality of this lower bound is highly dependent on the relative
sizes of the subgraphs. Ideally, we would like the size of the largest subgraphs
to be as small as possible, so we are forced to cut many of them to balance the
partitioning. Thus, we start a depth-first search from all vertices in $R$
simultaneously, updating each of the corresponding search trees one by one
(cycling through them using, for example, a queue) until all searches have
terminated.

\subsubsection{Combining bounds}
A priori, the flow and extended packing bound conflict with each other, so
if we compute both we would be relegated to taking the maximum of the two
(and then adding $|C|$). However we can first compute the flow bound, remove
the set of paths it found from the graph, and then run the packing bound on
the remainder of the graph. Then the packing bound will only use edges not used
by the flow bound and we can add the two bounds together.

Additionally while the combined flow and extended packing bound together should
give the strongest lower bound, they are also the most expensive to compute.
Therefore we also compute the weaker, local packing bound~\cite{pelt15}, as the
overhead to compute this is negligible. In particular, we do this computation
incrementally: if this local packing bound on its own already indicates this
subtree should be pruned, we do so without computing any other bounds. Then we
compute the flow bound and ask the same question, and then finally the extended
packing bound.

	\section{Experimental results}
\label{sec:experiments}
We implemented the branch-and-bound algorithm from \autoref{sec:opt} in our new program
MP\footnote{MP is available from \url{https://github.com/TimonKnigge/matrix-partitioner}.}. The implementation was
done in C++14, and the final program was compiled with GNU GCC Version 7.1.0 with the \texttt{-O2} flag. The
program was written sequentially, but as a branch-and-bound algorithm it can easily be parallelized in the future.

To test the capabilities of the new exact matrix bipartitioner MP and to compare it with MondriaanOpt, we performed
numerical experiments on a subset of small and medium-sized test matrices from the SuiteSparse Matrix Collection
(formerly known as the University of Florida Sparse Matrix Collection~\cite{davis11}). We chose as test set the
subset of all sparse matrices with at most 100,000 nonzeros, which contains 1602 matrices\footnote{After having removed five duplicate matrices:
	\texttt{Pothen/\-barth}, \texttt{Pothen/\-barth4}, \texttt{Mesz\-aros/\-fxm3\_6},
	\texttt{Boeing/\-nasa\-1824}, \texttt{Pajek/\-foot\-ball}.}\footnote{Retrieved September 2018.}.
We chose a value of $\epsilon=0.03$ in \autoref{eqn:imbal}, which is a common value in the literature
allowing a trade-off between load imbalance and communication volume. To keep the total CPU time used within
reasonable bounds, we allotted a maximum of 24 hours of CPU time to each partitioning run.

All computations were carried out on thin nodes with 24 cores of the Dutch national supercomputer Cartesius at
SURFsara in Amsterdam, with a core clock speed of 2.4 GHz (for Intel Ivy Bridge E5-2695 v2  CPUs) or 2.6 GHz (for
Haswell Bridge E5-2690 v3 CPUs). The memory for each thin node is 64 GB. Each batch of 24 jobs is assigned to a
node by a runtime scheduler, which may lead to different types of CPUs being used, causing a slight inconsistency
in our timings. The scheduler carried out all MP runs on the slower Ivy Bridge nodes and all MondriaanOpt nodes on
the faster Haswell nodes. Since this distribution favors MondriaanOpt, we provide all results and runtimes
uncorrected for this fact. For the sake of completeness, we performed calibration runs for the 40 longest running
matrices using MP on both types of nodes of the Cartesius computer, and based on the geometric mean for this set,
we found that the Haswell nodes are a factor of  $\alpha= 1.1782$ times faster than the Ivy Bridge nodes.

As a result of our experiments, we may divide the matrices into three groups: (i) a group of 368 matrices which
could be solved by both programs, MP and MondriaanOpt. We use these matrices to compare the speed of the two
programs and to verify their correctness; (ii) a group of 471 matrices which could only be solved by MP;
(iii) a group of 763 matrices which could not be solved by either program.
All matrices that could be solved by MondriaanOpt within 24 hours could also be solved by MP within 24 hours.

For the 368 matrices that could be solved by both programs, all optimal volumes computed are identical for the two
programs, which we take as an independent mutual confirmation of their correctness. We have taken great care
in developing the programs to make them understandable and to document them well, to support our claim that they
compute exact, optimal solutions. Both programs are available as open-source software and are open to inspection
for correctness. The two programs do not necessarily compute the same solution, as there may be several optimal
solutions. The optimal volume, however, is of course unique.

Of the 368 matrices solved by both programs, MP is faster than MondriaanOpt in 306 cases (83\% of the cases).
In 25 cases (6.7\%), it performs equally well, of which 22 cases with both programs needing exactly 1 second (our
clock resolution), and having a volume of 0. In 37 cases (10\%), MondriaanOpt is faster, of which 31 cases with
volume 0. For volume 0, the sparse matrix can be split into several connected components (when viewed as a graph)
of suitable sizes. This situation is easy to handle and it is quickly discovered by both programs. In general,
there is a trade\-off where the time spent computing sharper lower bounds should be compensated for by a sufficient
reduction in the size of the branch-and-bound tree. In our situation this appears to be the case in a large
majority of test matrices.

We synthesize the results of our experiments in \autoref{plotfig} below. With the exception of two matrices,
\texttt{mhd4800b} (29574s) and \texttt{mhd3200b} (7977s), all matrices solved by MondriaanOpt in a day could be
solved by MP in an hour, all but nine of them in fifteen minutes, and all but twenty in two minutes.
The geometric average of the ratio $T_{\mathrm{MP}} / T_{\mathrm{Opt}}$ between the time of MP and the time of
MondriaanOpt is $0.0855$. This average is based on 286 matrices that could be solved by both programs and for
which $T_{\mathrm{MP}}, T_{\mathrm{Opt}}  \geq 1$ second.
This means that on average, MP is more than ten times faster
than MondriaanOpt.

\begin{figure}[htb]
	\centering
	\def\svgwidth{0.5\textwidth}
	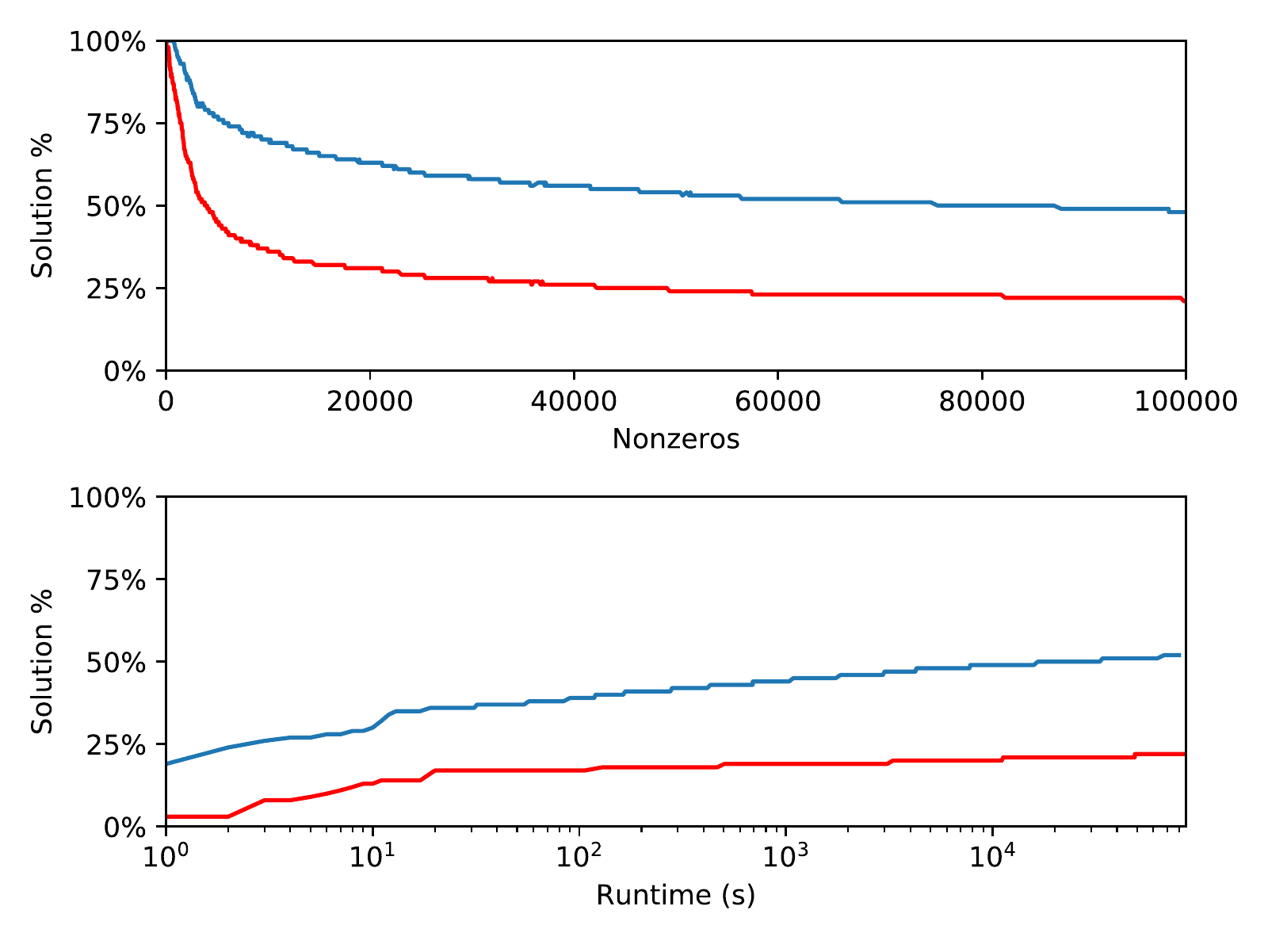
	\caption{In both plots, the blue graph is for \textbf{\color{blue}MP}, and red for
		\textbf{\color{red}MondriaanOpt}. The top plot shows for a
		given number of nonzeros $N$, what percentage of matrices with at most $N$ nonzeros were solved in
		the alloted time. The bottom plot shows for a given time in seconds $t$, the percentage of matrices
		in the whole test set that was solved in $t$ seconds. Note the log\-scale for $t$.}
	\label{plotfig}
\end{figure}

\autoref{tbl:matrices} shows the hardest (longest running) cases that MP could still solve within our self-imposed 
time limit.
These are in fact the matrices that take between 3 and 24 hours to be solved. Note that there is no simple
parameter that characterises the most difficult matrices for bipartitioning. Still we can say that the hardest
solvable matrices usually have a communication volume of 10 or more, with the exception of the matrix
\texttt{mhd4800b}, which has a low volume of 2. Furthermore, they also have at least 1000 nonzeros, with two
exceptions, \texttt{ch4-4-b2} and \texttt{GD97\_a}. This particular top-40 set is our challenge to future exact
partitioners.

In addition to the literature standard of $\epsilon=0.03$,
we ran this top 40 of matrices for $\epsilon = 0$ and $\epsilon=0.1$ as well, to see what the effect
of the choice of $\epsilon$ might have on these hard matrices. Results are also found in \autoref{tbl:matrices}.
We see that in many cases, the volume changes only slightly or not at all. There are several notable exceptions
however, where the larger imbalance allows for lower volumes which can also be found faster, for example for the
\texttt{iiasa} matrix. The comparative runtimes vary much more. Increasing $\epsilon$ will grow the size of the
solution space since more partitionings now satisfy the load balancing constraint. Conversely, this may allow us to
find lower volume solutions that we can use to prune larger parts of the search tree. Which of these effects is
stronger depends on the matrix, even in a non-linear fashion as can be seen from the results of \texttt{lp\_grow22}
where the choice of $\epsilon=0.03$ appears especially unfortunate\footnote{Do note the selection bias here.
Among the matrices that are very time consuming to solve for $\epsilon=0.1$ we might similarly find a matrix
solved very rapidly for $\epsilon=0.03$.}.

\begin{table*}[p]
\caption{The top-40 of matrices with the longest computation time needed by the matrix partitioner MP for
	$\epsilon=0.03$. Given are the matrix name, number of rows, columns, and nonzeros, and for each
	$\epsilon\in\{0,0.03,0.1\}$ the optimal communication volume and CPU time (in s) needed for computing
	an optimal solution. In case a matrix was not solved within the time limit of 24 hours, a dash (--)
	appears. Of these matrices, MondriaanOpt only solved \texttt{mhd4800b} for $\epsilon=0.03$ (in 55001s).}
\label{tbl:matrices}

\begin{center}
\begin{tabular}{lrrrrrrrrr}
	\hline
	Name&
	\multicolumn{1}{c}{$m$}&
	\multicolumn{1}{c}{$n$}&
	\multicolumn{1}{c}{$nz$}&
	\multicolumn{2}{c}{$\epsilon=0$}&
	\multicolumn{2}{c}{$\epsilon=0.03$}&
	\multicolumn{2}{c}{$\epsilon=0.1$}\\
	\cline{5-6}\cline{7-8}\cline{9-10}
	&&&&$V$&$T$&$V$&$T$&$V$&$T$\\
	\hline
	\hline

	\texttt{c-28}&4598&4598&30590&10&13275&10&11470&10&5494\\
	\texttt{mhd1280a}&1280&1280&47906&--&--&44&12375&--&--\\
	\texttt{bp\_1000}&822&822&4661&35&4058&35&13945&--&--\\
	\texttt{reorientation\_4}&2717&2717&33630&--&--&14&14680&14&28025\\
	\texttt{DK01R}&903&903&11766&--&--&20&15070&20&8527\\
	\texttt{west0479}&479&479&1910&33&20982&33&15156&33&46986\\
	\texttt{ch4-4-b2}&96&72&288&24&23469&24&15955&24&39888\\
	\texttt{celegansneural}&294&270&2345&58&10142&57&16203&--&--\\
	\texttt{lp\_stocfor3}&16675&23541&76473&14&13757&14&18138&14&37751\\
	\texttt{bayer02}&13935&13935&63679&28&18643&27&18832&27&30701\\
	\texttt{circuit204}&1020&1020&5883&--&--&41&19078&40&35546\\
	\texttt{orbitRaising\_4}&915&915&7790&--&--&16&19405&8&14\\
	\texttt{GD97\_a}&84&84&332&24&17958&24&20024&24&33897\\
	\texttt{lp\_modszk1}&686&1620&3168&34&60417&34&20768&34&45698\\
	\texttt{mhd4800b}&4800&4800&27520&2&23728&2&23298&0&12\\
	\texttt{Hamrle2}&5952&5952&22162&16&13270&16&24322&16&6156\\
	\texttt{dynamicSoaringProblem\_4}&3191&3191&36516&--&--&22&24533&22&42353\\
	\texttt{pcb1000}&1565&2820&20463&41&37923&40&24601&40&35283\\
	\texttt{lp\_grow22}&440&946&8252&20&340&20&25148&20&34\\
	\texttt{kineticBatchReactor\_5}&7641&7641&80767&--&--&18&26139&18&42805\\
	\texttt{can\_256}&256&256&2916&44&4193&43&28300&40&29983\\
	\texttt{qiulp}&1192&1900&4492&--&--&40&30445&--&--\\
	\texttt{ex21}&656&656&19144&--&--&62&31311&--&--\\
	\texttt{lp\_bnl1}&642&1586&5532&--&--&47&31663&--&--\\
	\texttt{c-29}&5033&5033&43731&--&--&28&32127&--&--\\
	\texttt{fs\_541\_1}&541&541&4285&37&23259&37&33046&--&--\\
	\texttt{fs\_541\_4}&541&541&4285&37&23247&37&33181&--&--\\
	\texttt{fs\_541\_2}&541&541&4285&37&23529&37&33264&--&--\\
	\texttt{fs\_541\_3}&541&541&4285&37&23303&37&33377&--&--\\
	\texttt{bp\_600}&822&822&4172&--&--&33&34244&--&--\\
	\texttt{model1}&362&798&3028&47&17281&46&37081&--&--\\
	\texttt{kineticBatchReactor\_9}&8115&8115&86183&--&--&18&44655&--&--\\
	\texttt{kineticBatchReactor\_4}&7105&7105&74869&--&--&18&45586&--&--\\
	\texttt{de063157}&936&1656&5119&36&8463&36&46148&--&--\\
	\texttt{ncvxqp9}&16554&16554&54040&30&33325&30&46556&30&64409\\
	\texttt{lp\_sctap2}&1090&2500&7334&41&29014&40&57920&--&--\\
	\texttt{iiasa}&669&3639&7317&14&11410&14&64331&6&3\\
	\texttt{can\_229}&229&229&1777&38&42965&38&65317&--&--\\
	\texttt{lp\_pilot4}&410&1123&5264&--&--&47&68010&44&688\\
	\texttt{lpi\_pilot4i}&410&1123&5264&--&--&47&70176&44&698\\

	\hline
\end{tabular}
\end{center}
\end{table*}

To emphasize the potential of MP as a benchmark for heuristic solutions, we applied Mondriaan 4.2.1
to the $839$ optimally solved matrices. We used both the underlying hypergraph
partitioner of Mondriaan as well as PaToH 3.2, and used the fine-grain strategy (with and without
iterative refinement) and medium-grain strategy (always uses iterative refinement). For all other
options we relied on the Mondriaan defaults. The average runtime (arithmetic mean) and optimality
ratio (geometric mean, after removing volume-0 matrices, at which point 726 matrices remain)
may be found in \autoref{heur}. Furthermore, \autoref{plotfig2} contains a performance plot of the
six configurations. To answer the question from the introduction `how good are the
current methods?', we note that the geometric average achieved by the combination of Mondriaan
medium-grain with iterative refinement and the PaToH bipartitioner is only 10\% above the optimal
attainable value, meaning that the heuristic partitioning is already very close to what can optimally
be achieved.

\begin{table}[htb!]
\begin{center}
	\resizebox{0.45\textwidth}{!}{
	\begin{tabular}{lcc}
		\hline
		Partitioner&Runtime (in s)&Optimality ratio\\
		\hline
		\hline
		Mondriaan FG	&0.051452	&1.63409\\
		Mondriaan FG+IR	&0.053945	&1.53390\\
		Mondriaan MG	&0.029856	&1.46326\\
		PaToH FG	&0.013915	&1.18534\\
		PaToH FG+IR	&0.015162	&1.16161\\
		PaToH MG	&0.009188	&1.10145\\
		\hline
	\end{tabular}}
\end{center}
	\caption{Results summary heuristic partitioners.}
	\label{heur}
\end{table}

\begin{figure}[htb]
	\centering
	\def\svgwidth{0.5\textwidth}
	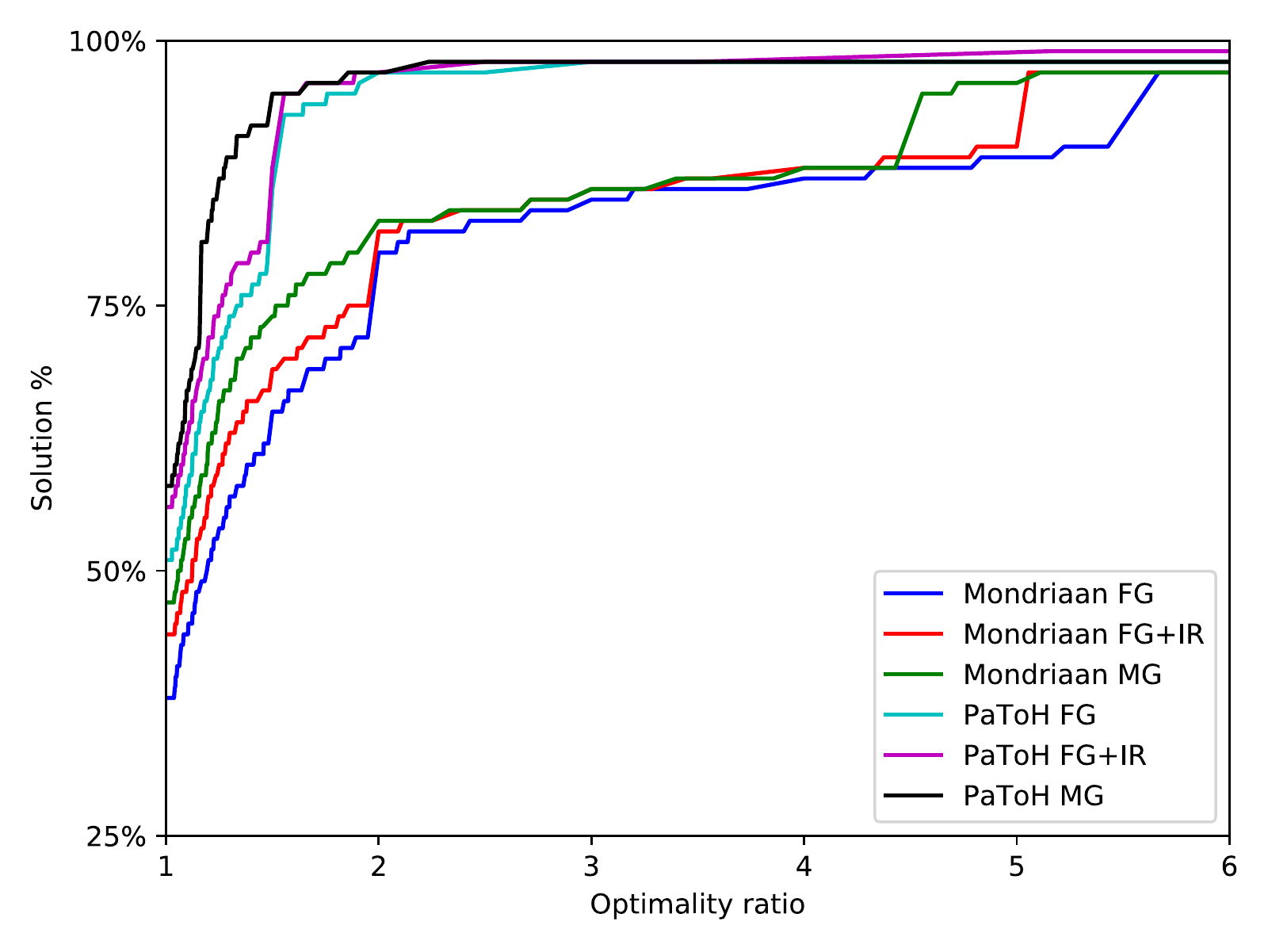
	\caption{For a given ratio $r \geq 1$, we show what percentage of the $839$ matrices
		were solved within a factor $r$ of the optimal solution found by MP, for
		\textbf{\color{blue}Mondriaan fine-grain}, \textbf{\color{red}Mondriaan fine-grain
		with iterative refinement}, \textbf{\color{mgreen}Mondriaan medium-grain},
		\textbf{\color{cyan}PaToH fine-grain}, \textbf{\color{purple}PaToH fine-grain with
		iterative refinement} and finally \textbf{PaToH medium-grain}. For matrices whose
		optimal volume was $0$, we used the convention that $x/0 = 1$ when $x=0$ and
		$\infty$ if $x \geq 1$.}
	\label{plotfig2}
\end{figure}

	\section{Conclusions and future work}
\label{sec:concl}
In this work, we have expanded our data base of 356 optimally bipartitioned sparse matrices to 839 matrices, by
developing a new flow-based bound and a stronger packing bound for our previous branch-and-bound
algorithm~\cite{pelt15}. We implemented this bound in a new matrix partitioner, MP, which has the same
functionality as the previous partitioner MondriaanOpt.
We are now able to bipartition 96.8\% of the sparse matrices with at most 1000 nonzeros from the SuiteSparse
collection~\cite{davis11} to optimality, reaching the exact minimum communication volume for a given load
imbalance $\epsilon=0.03$. For matrices with less than 10,000 nonzeros, we are successful in 72.8\% of the cases,
and for matrices with less than 100,000 nonzeros still in 52.3\%.
The new partitioner MP is more than ten times faster than MondriaanOpt for problems that both partitioners can
solve, and, more importantly, enables us to solve many more partitioning problems.

In the near future, we intend to apply the new partitioner also to selected problems that we could not solve within
our imposed limit of 24 hours. Looking already beyond the horizon, the smallest (by nonzeros) matrix that MP could 
not solve within our limit of one day is \texttt{cage6}. We partitioned this matrix using MP in 283,316 s (over 3
days) on a laptop computer with an Intel i7-8550U 1.8 GHz CPU. By comparison, one of the authors ran MondriaanOpt
for three months on this instance without solving it. The result of the 3-day calculation is shown in
\autoref{fig:cage6}.

In this paper, we also gave a proof of the \NP-completeness of $\epsilon$-balanced sparse matrix bipartitioning.
This result may hardly be surprising, as graph partitioning and hypergraph partitioning are both known to be
\NP-complete. Still, this problem is a very specific instance of hypergraph partitioning, and it is a motivation
for developing good heuristic partitioners to know that solving the problem optimally by an exact algorithm is
intractable. It is our hope that the expanded data base of optimally bipartitioned sparse matrices will be used in
practice to benchmark the quality of the bipartitioning kernel of such heuristic partitioners.

\begin{figure}[htb]
  \centering
  \frame{\includegraphics[width=8cm]{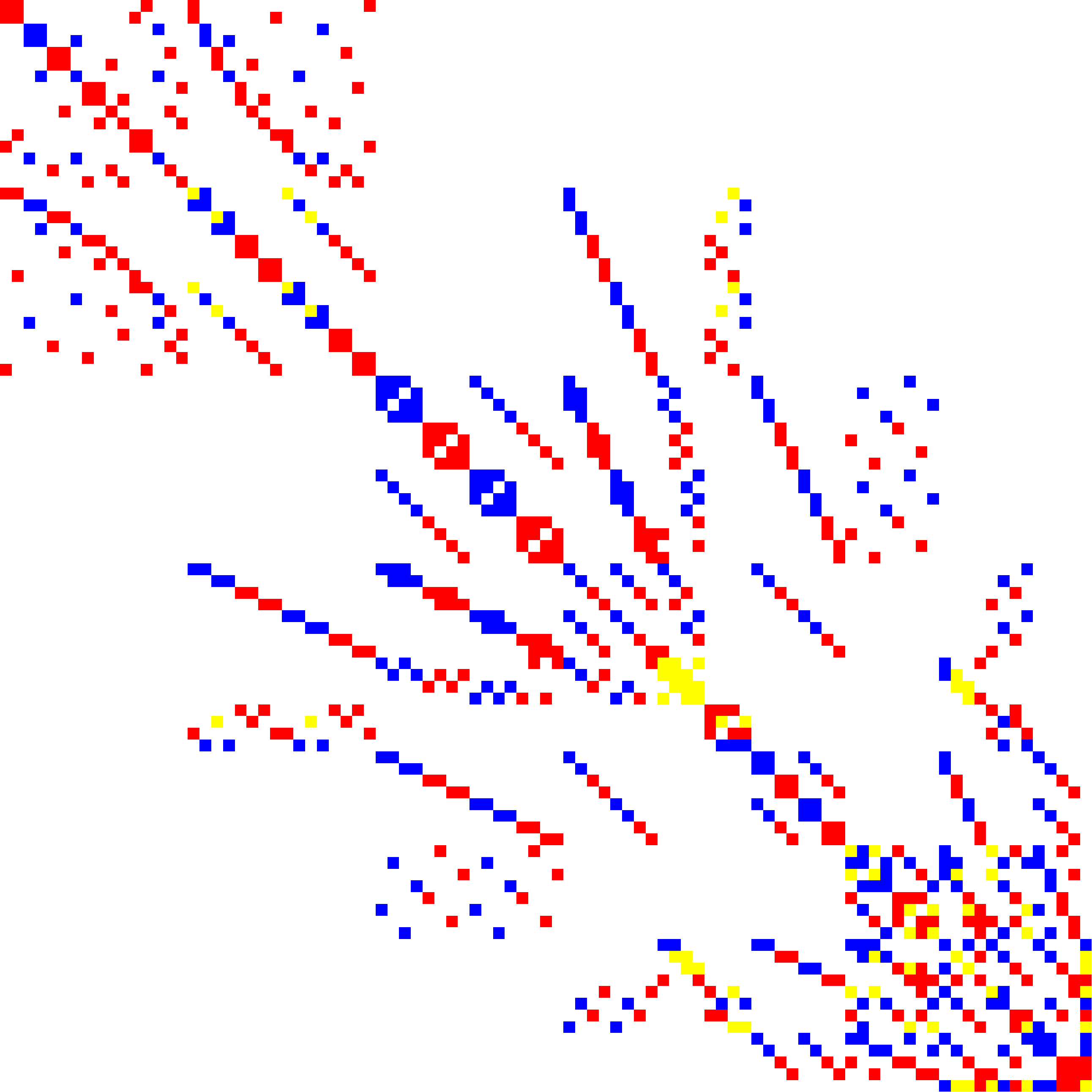}}
  \caption{Bipartitioning of the $93 \times 93 $ matrix \texttt{cage6} with 785 nonzeros.
The minimum communication volume for an allowed imbalance of $\epsilon=0.03$ equals $V=38$. 
The 397 red nonzeros are assigned to one part, the 316 blue nonzeros to the other part,
and the 72 yellow nonzeros can be freely assigned to any part without affecting
the communication volume, because both their row and their column is already cut.
We can color these free nonzeros blue to improve the load balance, giving 397 red and 388 blue nonzeros,
corresponding to an achieved imbalance of about $\epsilon^{\prime}=0.01$.}
\label{fig:cage6}
\end{figure} 

	\section{Acknowledgements}
We thank Oded Schwartz for helpful discussions on \NP-completeness.
The computations of this paper on the Dutch national supercomputer Cartesius
at SURFsara in Amsterdam  were carried out under grant
SH-349-15 from The Netherlands Organisation for Scientific Research (NWO).

	\bibliographystyle{elsarticle-num}
	\bibliography{improved}

\end{document}